\documentclass[epsf,11pt,times,psfrag]{article}
\usepackage[T1]{fontenc}
\usepackage{multirow,amssymb,amsmath,graphicx,epsfig,graphics,geometry,setspace}
\usepackage{wasysym,hyperref,textcomp}
\usepackage{graphicx}
\usepackage{ulem}
\usepackage[usenames]{xcolor}
\usepackage[noend]{algpseudocode}
\usepackage{algorithm}

\newcommand{\ignore}[1]{{}}
\newcommand{\twonormsq}[1]{{||{#1}||}^2_2}
\newcommand{\twonorm}[1]{{||{#1}||}_2}
\newcommand{\E}{\mathbb{E}}

\newcommand{\plusminus}{\pm}

\newenvironment{proof}{\par\noindent{\bf Proof:}}{\mbox{}\hfill$\Box$\\}

\newtheorem{theorem}{Theorem}[section]

\newtheorem{lemma}{Lemma}[section]

\newtheorem{observ}{Observation}[section]
\newtheorem{cor}{Corollary}[lemma]
\newtheorem{claim}{Claim}

\begin{document}
\title{Randomized Rounding Revisited with Applications
}
\author{ 
Dhiraj Madan\thanks{Email:\texttt{dhirajmadan1@gmail.com}} \ and \ 
Sandeep Sen\thanks{Email:\texttt{ssen@cse.iitd.ernet.in}} \\
   Department of CSE,\\
   I.I.T. Delhi, India
}  

\maketitle
\begin{abstract}
We develop new techniques for rounding packing
integer programs using iterative randomized rounding.
It is based on a novel application of multidimensional Brownian motion
in $\mathbb{R}^n$.
Let $\overset{\sim}{x} \in {[0,1]}^n$ be a fractional feasible solution of a
packing constraint $A x \leq 1,\ \ $ $A \in {\{0,1 \}}^{m\times n}$
that maximizes a linear objective function.
The independent randomized rounding method of Raghavan-Thompson rounds
each variable $x_i$ to 1 with probability $\overset{\sim}{x_i}$ and 0 otherwise.
The expected value of the rounded objective function matches the fractional
optimum and no constraint is violated by more than $O(\frac{\log m}
{\log\log m})$.
In contrast, our algorithm iteratively transforms
$\overset{\sim}{x}$ to $\hat{x} \in {\{ 0,1\}}^{n}$ using a random walk,
such that the expected values of $\hat{x}_i$'s are consistent with the
Raghavan-Thompson rounding. In addition, it gives us intermediate
values $x'$ which can then be used to bias the rounding towards a superior
solution.  

The reduced dependencies between the constraints of the sparser
system can be exploited using {\it Lovasz Local Lemma}.
Using the Moser-Tardos' constructive version,
$x'$ converges to $\hat{x}$ in polynomial time to a distribution over the
unit hypercube ${\cal H}_n = {\{0,1 \}}^n$ such that the expected value of any
linear objective function over ${\cal H}_n$ equals the value at $\overset{\sim}{x}$.

For $m$ randomly chosen packing constraints in $n$
variables, with $k$ variables in each inequality,
the constraints are satisfied within
$O(\frac{\log (mkp\log m/n) }{\log\log (mkp\log m/n)})$ with high probability where
$p$ is the ratio between the maximum and minimum coefficients of the linear objective
function.  For example,
when $m,k  = \sqrt{n}$ and $p = polylog (n)$, this yields $O(\log\log n/
\log\log\log n)$ error for polylogarithmic weighted objective functions that
significantly improves the $O(\frac{\log m}{\log\log m})$ error incurred by
the classical
randomized rounding method of Raghavan and Thompson \cite{RT:87}.

Further, we explore trade-offs between approximation factors and error,
and present applications to well-known problems like circuit-switching,
maximum independent set of rectangles and hypergraph $b$-matching.
Our methods apply to the weighted instances of the problems and are likely
to lead to better insights for even dependent rounding.

\end{abstract}
\newpage
\section{Introduction}
Many combinatorial optimization problems can be modeled
using a weighted packing integer program.
$ \mbox{ Maximize } \sum_{i=1}^{i=n} c_i \cdot x_i$
subject to
$ \sum_{ i \in S_j } x_i \leq 1 \ \ 1 \leq j \leq m \ \ \ 
 x_i \in \{0, 1 \}$.\\ 
Wlog, we can assume that $\max_i c_i = 1$.
Although the above formulation appears somewhat restrictive
having a fixed right hand side of each inequality, the
methods that we develop extend to more general parameters.
Often the constraints are expressed as 
$ C_j: V_j \cdot x \leq 1$ 
where $V_j$ is a 0-1 incidence vector corresponding to set $S_j \subset
\{ 1, 2\ldots n \}$. We will also use $x_i$ to denote the $i$-th coordinate
of a vector $x$.
Since this version is also NP-hard as it captures many intractable
independent set problems,
a common strategy is to solve the Linear Program (LP)
 corresponding to the relaxation $0 \leq x_i \leq 1$. Then
use the LP optimum $OPT$ to obtain a good approximation of the optimal
integral solution. 
Suppose the optimum is achieved by the
vector $x' = [ x'_1 , x'_2 \ldots x'_n ]$ 
In the conventional {\it randomized rounding} \cite{RT:87}, we 
round each $x_i$ independently in the following manner.\\ 
$\hat{x}_i =
\{ 1 \text{ with probability $x'_i$ and } 
0 \text{ otherwise} \}$.
\\
Since $\E[ \hat{x}_i ] = x'_i$, it follows that $\E [ \sum_i c_i \hat{x}_i ] 
=  \sum_i c_i \cdot x'_i $ and using Chernoff bounds, we can show that for
all $j$, w.h.p. $ \sum_{ i \in S_j } \hat{x}_i 
\leq \left(\frac{\log m}{\log\log m}
\right) $.

The primary motivation for improved randomized rounding
is improved approximation of
integer linear packing problems. Historically, the original randomized
rounding technique was proposed for obtaining better approximation of 
multicommodity
flows \cite{RT:87}. 
Srinivasan \cite{AS:95,AS:99} presents an extensive survey of many 
sophisticated variations of
randomized rounding techniques and applications to approximation algorithms.
It was established much later \cite{CGKT:07}
that the Raghavan-Thompson bound cannot be improved as a consequence of the
following result.
\begin{lemma}[\cite{CGKT:07}]
There exists a constant $\delta > 0$ such that the integrality gap of the 
multicommodity flow relaxation problem is $\Omega ( 1/c' \cdot 
n^{\frac{1}{3c' + 13}} )$ for any congestion $c'$, $1 \leq c' \leq 
(\delta \log n)/
\log\log n $ where $n$ is the number of vertices in the graph and the 
integrality gap for $c'$ is with respect to congestion 1. 
\end{lemma}

Therefore no rounding algorithm can simultaneously achieve error
$c' = o( \log n/\log\log n)$, 
and an objective function value of $\Omega ( OPT )$ for all
polynomial size input ($m \leq n^{O(1)}$). 
Note that this does not preclude improvement of the multicommodity flow problem
approximation by alternate formulation - for this, there are alternate 
hardness bounds given by \cite{CGKT:07}. 

Subsequent to the work of Raghavan-Thompson, Srinivasan \cite{AS:95,AS:96}
Baveja and Srinivasan \cite{BS:00}, Kolliopoulous and Stein \cite{KS:98}
obtained results that focussed on circumventing the above bottleneck
by making use of some special properties of the constraint matrices like
column-restricted matrices. In 
particular, the focus shifted to designing approximation algorithms that
are feasible (unlike the basic independent rounding of RT) by bounding
dependencies between inequalities and using {\it Lovasz Local Lemma} (LLL)
or even more sophisticated corelation inequalites like FKG \cite{AS:96}.
Leighton, Rao and Srinivasan \cite{LRS:98} made very clever use of
LLL to achieve nearly optimal results in graphs with {\it short} flow paths.
The notion of short flow paths was developed further in the context of 
the {\it unsplittable} 
multicommodity flow problem with an explicit objective function in many
papers including \cite{KR:96,BS:00,KS:06,CCGK:07}. Intuitively, short
flow paths reduce dependencies between conflicting routing
paths and makes the algorithms more efficient in terms of assigning paths
without exceeding the maximum allowed congestion (for edge-disjoint paths
it is 1). 

In this paper, we attempt to provide a uniform framework behind many prior
uses of independent rounding by recasting the process as
a Brownian walk in high dimension and analyzing its convergence. For readers
familiar with this technique, it may appear to be an overkill since the
final outcome of Brownian motion and independent rounding are identical 
(we shall formalize this in the next section). However, we will demonstrate
that this framework offers a cleaner exposition and simpler explanations 
of many of the previous clever, yet {\it ad-hoc} techniques. 
We have presented
a detailed characterization of the time-dependent behavior of the random variables 
associated with each inequality executing Brownian motion which was not studied before
to the best of our knowledge.   
It could potentially offer new tools for analysis of more complicated rounding
methods where the random variables may not be independent. 

We demonstrate two distinct applications of our new framework
by addressing the rounding problem for the {\it average} case of
an input family that we will define more precisely. 
\\
(i) A simple approximation algorithm for {\it weighted} objective
function. 
\\
(ii) An improvement of the Raghavan-Thompson rounding error for
a restricted class of weighted objective function.

As opposed to
the one-shot rounding of \cite{RT:87} , also referred as {\it independent
rounding}, our rounding is iterative and based on successive refinements
of the LP solution that has a lower variance per step, to yield a sparser
set of constraints. This reduces dependencies between inequalities
and makes the use of techniques like Lovasz Local Lemma (LLL) more effective. 
\subsection{Main results and some applications}
\begin{theorem}
Given an $m \times n$ 0-1 matrix $A$, such that $A \cdot x' \leq b$ for 
$x'_i \in [0,1]$,
and an objective function $\sum_i c_i x_i$, 
in randomized polynomial time $x'$ can be transformed into $y'$ which 
are $[0,1]$ valued random variables such that for all $(\log n - \log\log m 
+ \log b ) \geq p \geq 0$\\
(i) All constraints have $ \leq \frac{nB}{2^p}$ non-zero variables from 
$y'$ and
\\
(ii) For all $i$, $\sum_i A_i \cdot y' \leq O(\sqrt{\frac{2^p B \log m}{n}})$ where 
$A_i$ is the $i$-th row of $A$.
\\
(iii) $\E [ c_i y'_i ] = \sum_i c_i x_i$
\label{sparsify}
\end{theorem}
In other words we have a sparse system of equivalent constraints with the objective
functional value unchanged in the expection.
\\

{\bf Remark} For $p = \log n - \log\log m + \log B$, there are $\log m$ 
variables in each inequality summing upto $O(1)$. This specific result has
been observed before by using a clever application of independent rounding
in the following manner. Round each $x'_i < \frac{1}{\log m}$ to 
$\frac{1}{\log m}$ with probability $x'_i \log m$. Then $\E [x_i] =
x'_i$ and no inequality has more than $\log m$ variables with high probability.
\ignore{
Our transformation does not use any extra variables unlike the 
previous methods that makes multiple copies of each variable 
(see for example \cite{CC:09}). 
They transform the initial matrix $A$ to an $m\times n\log n$ matrix $A'$ and 
$x'$ is transformed into $y'$ that has dimension $n\log n$, using 
another application of independent rounding (some papers present it as 
a technique to do $\frac{1}{\log n}$ rounding) so
that the fractional solution remains feasible. Subsequently,
care has to be taken to ensure that the rounded solution chooses exactly
one copy of a variable which is an implicit case of dependent rounding. 
The {\it path decomposition} strategy for multicommodity flow 
in the original paper of Raghavan and Thompson \cite{RT:87} is also 
along similar lines.  While it works for specific instances, it will be
technically challenging to formalize it as a {\it blackbox} 
technique similar to the previous theorem.    
}

Using the above sparsification, we obtain a number of interesting 
results arguably in simpler ways than known previously.  
\ignore{
\begin{theorem}
Let $A \in {\{0,1\}}^{m \times n}$, with no more than $\rho$
1's per column.
Suppose $x' \in {[0,1]}^n$ maximizes $\sum_i c_i x_i $ and
satisfies $ A \cdot x \leq \cdot e^m $ with $OPT = \sum_i c_i x'_i$.
Then $x'$ can be rounded to $\hat{x} \in {\{0,1 }\}^n$ in polynomial time 
such that
the following holds with high probability, 
\begin{eqnarray}
 {|| A \hat{x}||}_\infty & \leq & 1
\text{ and } \\
\sum_i c_i \hat{x}_i & \geq & \Omega ( OPT/(\rho \log m) )
\end{eqnarray}
\label{mainthm2}
\end{theorem}
\vspace{-0.3in}
The above result follows by applying a simple greedy procedure to the 
sparsified matrix.

For a more general RHS, $b \geq 1$, we use LLL-based techniques to obtain the following
results.
}
\begin{theorem}
Let ${\cal A}^{m \times n}_k$ denote the family of $m \times n$ 
matrices where each row (independently) has $k$ ones in randomly 
chosen columns from $\{ 1, 2, \ldots n \}$ and 0 elsewhere\footnote{Alternately
we can set every entry to be 1 with probability $k/n$ independently}. 
Let $A \in 
{\cal A}^{m \times n}_k$, then
for any point $x' = ( x'_1 , x'_2 \ldots x'_n ) , \
0 \leq x'_i \leq 1$ such that
$ A \cdot x' \leq e^m $, where $e^m$ is a vector of $m$ 1's,
and $OPT = \sum_i c_i \cdot x'_i$ where 
$1 = \max_i \{ c_i \}$ and
$\forall i\ c_i \geq \frac{1}{p}$. Then,
$\bar{x'}$ can be rounded to $\hat{x} \in {\{0,1 }\}^n$ such that
the following holds with probability $\geq 1 - \frac{1}{m}$
\begin{eqnarray}
{|| A \hat{x}||}_\infty & \leq & O\left(\frac{\log(mkp\log m/n) + \log\log m }{\log\log (mkp\log m/n + \log m)}\right)
\text{ and } \label{cond1}\\
\sum_i c_i \hat{x}_i & \geq & \Omega ( OPT )
\end{eqnarray}
Moreover such an $\hat{x}$ can be computed in randomized polynomial time.
\label{mainthm0-1}
\end{theorem} 
{\bf Remark} 
(i) The result relates the rounding error to the average number of ones
in a column, i.e., $\frac{mk}{n}$ and characterizes tradeoffs between the
parameters $m, n, k$. For example, for $k = \sqrt{n}$, we can bound the error
to $O(\log\log n)$ for $m \leq \sqrt{n} \cdot polylog(n)$ for
$p  \leq \frac{1}{\log^{O(1)} n}$.
This is a significant improvement over the $O(\frac{\log n}{\log\log n})$
bound of the independent rounding. 
\\
(ii) For $k = \log n$, the bound is indeed tight.
A proof is given in the appendix (\cite{V:15}).
\\
(iii) For the unweighted case, the result holds for arbitrary distribution 
of the $k$ 1's in each row.\footnote{ 
Set each variable to 1 with
probability $1/k$, and for each violated constraint that has more than $q =
\log ((mk)/n )$
ones, zero all its variables (or zero enough of its variables, chosen
arbitrarily, so that it has only $q$ ones). At most $n/2$
of the variables are heavy in the sense that they appear in
more than $2mk/n$ constraints. For a light variable, even if it comes up 1, the
probability that it is a member of a violated constraint is small (say, below
1/2), and hence the light variables (using linearity of expectation) guarantee
a value of $\Omega(n/k)$.\\
This proof sketch was given by an anonymous reviewer of an earlier version.
}
\begin{theorem}
Let $A \in {\cal A}^{m \times n}_k$,
such that  $x' \in {[0,1]}^n$ maximizes $\sum_i c_i x_i $ and
satisfies $ A \cdot x \leq b \cdot e^m $ with $OPT = \sum_i c_i x'_i$.
If $1 = \max_i c_i \geq 1/p ,\ \ p \geq 1$ 
and $m \leq \frac{n\log m}{k}$,
then $x'$ can be rounded to $\hat{x} \in {\{0,1 }\}^n$ in polynomial time 
such that the following holds with high probability for $b \geq 1$
\begin{eqnarray}
 {|| A \hat{x}||}_\infty & \leq & b
\text{ and } \\
\sum_i c_i \hat{x}_i & \geq & \Omega ( OPT/(max(p \log m ,\log^2(m))^{1/b})) 
\end{eqnarray}
\label{mainthm3}
\end{theorem}
\vspace{-0.3in}
\color{black}
We observe a trade-off between the
weights and the approximation factor ${(p \log m ))}^{1/b}$. For $b = O(
\frac{\log\log m}{\log\log\log m})$, we can obtain $O(1)$ approximation for 
for $c_i \geq \frac{1}{\log^C n}$ for any constant $C \geq 1$. 

The random matrices provides a natural framework for combinatorial
results related to random hypergraphs. We sketch one such application to
$b$-matching of $k$-regular hypergraphs that yields 
an approximation $k^{1/b}$ for $b \geq 2$ using the result of Theorem
\ref{mainthm3} and
extends a result of Srinivasan \cite{AS:96}. This implies that 
for $b = \Omega ( \log\log n/\log\log\log n )$, we can can obtain a
$b$ matching of
size $\Omega (OPT)$ for {\it most} hypergraphs which matches the best
possible size given by the fractional optimum. 
It is known that $k$-uniform $b$-matching problem cannot be approximated
better than $\frac{k}{b\log k}$ for $b \leq k/\log k$ unless
$P = NP$ \cite{OFS:11,HSS:06}. 

\ignore{For lack of space, we have sketched details of the other applications
in the Appendix, section \ref{sec6}.}
\subsection{An overview and related work}
Our algorithm can be best characterized as a randomized iterated rounding where we begin
from a fractional feasible (specifically optimal) solution $x'$ and iteratively
converge to a {\rm good} integral solution. Our algorithm
has two distinct stages - {\it random walk} stage (more precisely, Brownian walk) and 
subsequently in the second stage invokes the Moser-Tardos iterative scheme 
for constructive Lovasz
Local Lemma (LLL). In the first stage we effectively 
{\it slow down} the RT rounding 
process. 
Our approach is intuitive  - starting from $x'$, for each
variable (dimension), we will roughly increment 
(actually a normal Gaussian increment)
$x_i$ by $\plusminus \gamma$ for
a suitably chosen $1 > \gamma > 0$ where the sign (direction) is a
random variable. In each iteration,
the values of $x_i$'s are modified and we continue this
process for each $x_i$ until it is in the range $[0, \delta ] \cup
[1- \delta,1]$ for an appropriate $1 > \delta > 0$ such that
$\delta > \gamma$. 
At this point we {\it fix} the variable
and we terminate when all inequalities have less than some predetermined 
value $u$ of {\it unfixed} variables . 
This stage has some similarities with the method 
of Lovette and Meka \cite{LM:12} but our analysis requires completely 
different techniques.
The crux of the method called {\it partial coloring lemma}
is a rounding strategy of an arbitrary $ x \in {[-1 , +1]}^n$ vector within
the discrepancy polytope defined by the constraints starting with
$x = (0,0, \ldots 0 )$.
Their method can be mapped to
$\{ 0, 1\}$ rounding as well, that was observed by Rothvoss\cite{R:13}.
Compared to the setting of the discrepacy rounding, we are dealing with
smaller error margins and the variable and polytope constraints are
not widely separated in terms of distances from the starting point.
Moreover, one needs
to also account for the deviation of the objective function which is
not required for Spencer's discrepancy result.

\ignore{
At a high level, our framework provides an alternate 
technique for limiting dependencies. In the remark following Theorem 
\ref{mainthm0-1} we noted that the rounding error is related to the average
number of 1's in a column which is related to the path-lengths in switching
circuits (c.f. section \ref{sec6}). 
We show that $\Omega (\frac{n \log\log n}{ \log\log\log n})$ 
connections can be supported in a 
{\it multi-butterfly} network with congestion
bounded by $O(\log\log n/\log\log\log n )$, extending a result of 
\cite{MS:99,CMHMRSSV:98} where they achieve a similar congestion for $n$ flows. 

We also present approximation algorithms for the Maximum Independent Set of
Rectangles (MISR) problem where the rectangles are aligned with grid points.
The rectangles can have arbitrary aspect ratios but their areas are bounded.
This version can be useful for many applications including map labelling
where the bounding boxes are not arbitrarily large and do not intersect 
too many cliques. For the case that the rectangles contain at most
polylog grid points, we present a
$O(\frac{\log\log^2 n}{\log\log\log^2 n})$ approximation bound for the 
weighted version. Obtaining an approximation factor $o(\frac{\log n}
{\log\log n})$ 
been an important open problem (\cite{CC:09,CH:09}) and has 
partly motivated the recent work in 
quasi-polynomial time approximation algorithms \cite{AW:13}.   

Further, the random matrices provides a natural framework for combinatorial
results related to random hypergraphs. We sketch one such application to
$b$-matching of $k$-regular hypergraphs that yields 
an approximation $k^{1/b}$ for $b \geq 2$ using the result of Theorem
\ref{mainthm2} and
extends a result of Srinivasan \cite{AS:96}. This implies that 
for $b = \Omega ( \log\log n/\log\log\log n )$, we can can obtain a
$b$ matching of
size $\Omega (OPT)$ for {\it most} hypergraphs which matches the best
possible size given by the fractional optimum. 
It is known that $k$-uniform $b$-matching problem cannot be approximated
better than $\frac{k}{b\log k}$ for $b \leq k/\log k$ unless
$P = NP$ \cite{OFS:11,HSS:06}. 

For lack of space, we have sketched details of the above applications
in the Appendix, section \ref{sec6}.
\subsection{Prior related work in randomized rounding}

Subsequent to the work of Raghavan-Thompson, Srinivasan \cite{AS:95,AS:96}
Baveja and Srinivasan \cite{BS:00}, Kolliopoulous and Stein \cite{KS:98}
obtained results that are similar in spirit to this paper. A direct
comparison of these results is a complex exercise 
since they had not addressed this specific framework of random matrices. 
Later, we
will be able to compare 
them with respect to our result on column-restricted 
0-1 matrices that we use as an intermediate step
to prove the main results. Although, this intermediate
result falls a little short of the best known, we feel that our techniques
are simpler and more general that could open up a new paradigm for rounding 
more general integer programs.

There exists a very extensive and rich literature on the use of randomized
rounding for solving various generalizations 
of the edge-disjoint path problem starting
with the seminal paper of Raghavan and Thompson. Leighton and Rao \cite{LR:88,
LR:99} formalized many algorithmic an combinatorial 
aspects of the multicommodity flow problem. 
The problem of edge
disjoint paths in expander graphs was resolved by Bohman and Frieze
\cite{BF:01}. 


The recent work on constructive discrepancy perhaps comes
closest to our technique of randomized iterative rounding. 
Bansal's \cite{bansal:10} seminal
paper on a constructive proof of Spencer's discrepancy theorem is based
on rounding solutions of successive
a semi-definite program that captures the discrepancy constraints.
This method was further refined and simplified in the work of
Lovette and Meka\cite{LM:12} who
derived an alternate proof
based on a very elegant analysis of a multidimensional random walk.
The crux of the method called {\it partial coloring lemma}
is a rounding strategy of an arbitrary $ x \in {[-1 , +1]}^n$ vector within
the discrepancy polytope defined by the constraints starting with
$x = (0,0, \ldots 0 )$. 
Their method can be mapped to
$\{ 0, 1\}$ rounding as well, that was observed by Rothvoss\cite{R:13}.
Compared to the setting of the discrepacy rounding, we are dealing with
smaller error margins and the variable and polytope constraints are 
not widely separated in terms of distances from the starting point. 
Moreover, one needs
to also account for the deviation of the objective function which is
not required for Spencer's discrepancy result.
}

\section{A Rounding procedure using random walks}
\label{sec:algo}
The simple random walk based algorithm outlined in the introduction 
doesn't take into account any of the constraints $V_j 
\cdot x \leq 1$ and therefore likely to violate them after some
random walk steps. However, the probability that an $x_i$ reaches
1 before it reaches 0 is equal to the ratio 
$\frac{x'_i /\gamma}{1- x'_i + x'_i /\gamma}
= x'_i$. 
This is
a consequence of a more general property of martingales known as Doob's
optional stopping theorem  (\cite{feller1}).
\begin{theorem}
Let $(\Omega , \Sigma, P)$ be a probability space and $\{ F_i \}$ be a 
filteration of $\Omega$, and $X = \{ X_i \}$ a martingale with respect to
$\{ F_i \}$. Let $T$ be a stopping time such that $\forall \omega \in \Omega, 
\ \forall i , \ \ | X_i (\omega ) | < K$ 
for some positive integer $K$, and $T$ is almost surely bounded.
Then $\E [ X_T ] = \E [ X_0 ]$.   
\label{stopping-theorem}
\end{theorem} 
In the above application, the stopping times are $X_T = \{-b , a \}$ whichever
is earlier. So $-b \cdot \Pr [ X = -b ] + a \Pr [ X =a] = 0 \cdot \Pr[ X = 0]$.
Since $\Pr [X = -b] + \Pr [ X=a ] = 1$ from the stopping criteria, 
$\Pr [ X= a] = \frac{b}{a+b}$. In our context, we start from $X = x'_i$
and $b = \frac{x'_i}{\gamma}$ and $a = \frac{1- x'_i}{\gamma}$, so
the probability that $\hat{x}_i = 1$ equals $\frac{x'_i /\gamma}
{(1 - x'_i + x'_i )/\gamma} = x'_i$. 
So the distributions of the variables 
being absorbed at 0 or 1 are identical to the independent rounding. 
However, the random walk process has many intermediate steps that do not
have corresponding mappings in the $2^n$ possible configurations of the
independent rounding. Whether it makes this framework more powerful compared
to one-step independent rounding is a difficult question but we feel that 
it could
give us a superior understanding of the process of independent rounding.   

In the algorithm presented in Figure \ref{algo1}, instead of the simple
Bernoulli random walk step, we use a normal Gaussian random walk on
each coordinate. 
We run the basic algorithm for
$T$ iterations such that the number of un-fixed variables in each 
constraint of $C^T$ (the active set of constraints after $T$ iterations) 
is bounded by $\log n$. 
The value of $T$, more specifically $\E[ T ]$, will be
determined during the course of our analysis. 

 Now we run the algorithm of \cite{MT:10} (c.f. section \ref{sec5}) 
on the constraints in $C^T$ which
are projections of the original constraints on the unfixed variables after
$T$ steps.
\begin{figure}
\fbox{\parbox{6.0in}{
\begin{center}
Algorithm {\bf Iterative Randomized Rounding }
\end{center}
{\footnotesize
Input : $x'_i \ \ 1 \leq i \leq n$ satisfying $Ax \leq b \cdot e^m$. $C^0$: set of constraints.
\\
$0 < \gamma < \delta < 1$ - the exact values are discussed in the analysis.
\\
Output : $\hat{x_i} \in \{0, 1\}$ \\

Initialize all variables as {\it un-fixed}
and set $X_0 = [x'_1 , x'_2 , \ldots x'_n ] $.\\
{\bf Repeat} for iterations $i \ \  = 1, 2 \ldots $
\begin{quote}
\begin{enumerate}
\item Generate a random vector $ R^{i} = U_i$ where $U_i$ is
a multidimensional gaussian r.v. restricted to the unfixed variables.
\item $X_{i+1} = X_i + \gamma \cdot R^i$.
\item {\it Fix} a variable if it is
less than $\delta$ or greater than $1-\delta$.\\ 
If all variables are fixed then exit. \\
\{ * multiple variables may get fixed in a single iteration. * \}
\item Update the set of constraints $C^t$ that contains at least one
unabsorbed variable.
\end{enumerate}

\end{quote}
{\bf until} stopping condition ${\cal S}$ (* 
$\max_j \{ \mbox{ number of unfixed variables in } C_j \}  \leq \log n $ *)
\\[0.1in]
Run Moser-Tardos \cite{MT:10}  algorithm on $C^T$ on the unfixed 
variables of $X_T$ according
given in Figure \ref{algo2} \\[0.1in]
Round the {\it fixed} variables to 0 or 1 whichever is closer 
 and return this vector denoted by $\hat{x}$.
}
}}
\caption{An iterative randomized rounding algorithm based on random walks}
\label{algo1}
\end{figure}
\subsection{The framework and some notations }
In the rounding algorithm, $X_t$ denotes
the random walk vector after $t$ steps. 
Let $U^d$ denote a
 $d$-dimensional Gaussian random variable
and let $U_t$
denote the projection of $U^n$ on the current subspace corresponding to
the unfixed variables in the
$t$th iteration. 
The normal distribution is
denoted by ${\cal N}(\mu , \sigma^2)$ where $\mu$ is the
mean and $\sigma^2$ is the variance. 
Unless otherwise stated, we will
refer to the {\it standard} normal distribution where $\mu = 0$ and
$\sigma^2 = 1$. 

The parameters $\gamma, \delta$ are chosen to ensure that the walk stays within
the feasible region. It suffices to have $\gamma \leq \frac{\delta}{\log n}$
from the pdf of the normal distribution if we are executing $O(\frac{1}{
\gamma^2})$ steps (see \cite{LM:12} for a rigorous proof). The value of
$\delta$ will be determined according to the approximation factor and the
error bound that we have in a specific application. In our analysis, 
the the focus will be on 
variables absorbed at 0 that will be rounded down. The decrease in
the objective value can be at most $n\delta$ (recall the maximum weight 
of the coefficients is 1). If we choose $\delta$ such that $n\delta$ is
$o(OPT)$ it will suffice. For the main theorems, choosing $\delta \leq
\frac{1}{polylog n}$ will work.

After $t$ iterations, the $j$-th constraint is denoted by
$C_j^t : < V_j , X_t > \ \leq 1 + \beta_t$ that has many of variables
{\it fixed}. We shall denote the set of unfixed variables in iteration $t$ 
by ${\cal U}(t) \subset
\{1, 2 \ldots n\}$ and the corresponding vector by $X^{{\cal U}}_t$ where the
the fixed variables are set to 0. The constraint vector $V^t_j =
V_j \cap {\cal U}(t)$, where only the coefficients in $V^t_j$ corresponding to 
${\cal U}(t)$ can be non-zero. Then\\
$|< V^t_j , X_t - X_{t+1}>| = |< \frac{V_j}{\twonorm{ V_j }} , 
X^{{\cal U}}_t - X_{t+1}^{{\cal U}}>| \cdot
\twonorm{ V_j }$ is the change in the value of $C_j$ in
iteration $t$ as the remaining
variables do not change. 
Note that while $X^{\cal U}$ changes in every step, $V^t_j$ changes only when
some variable is absorbed. 

For convenience of the analysis, we will club successive iterations 
into {\it phases}, where
within a phase $p$, $\beta_p$ remains unchanged. Equivalently, $\beta_p 
- \beta_{p-1}$
reflects the cumulative effect of a number of random walk steps within the 
phase $p$ referred to as the {\it accumulated error} or simply {\it error}.  
The phase $p$ corresponds to the $L_2$ norm of any
constraint $\twonorm{V^{p}_j}$ that is bounded by $\sqrt{n/2^p}$. 
Intuitively, with additional random walk steps, we are more 
likely to violate the original constraints and $\beta_p$ is a measure of
the violation. 
In terms of the above notation, it is obvious that
$ \twonorm{V^{p+1}_j} \leq \twonorm{V^{p}_j}$. We will use the index $p$ (
respectively $t$) to indicate reference to phases (resp. iterations).

Wlog, we assume that all constraints have at least 2 
variables and at most $n-1$
variables. \footnote{A constraint with $n$ variables has a trivial solution
where any one variable can be set to 1 and a constraint with one
variable is redundant.} 
The successive Brownian motion steps (defined by multidimensional normal
Gaussian) form a martingale sequence.
The following result forms the crux of our analysis - see \cite{bansal:10,LM:12}
for more details and proof.
\begin{observ}
In iteration $t$, let $Y_t = < V_j , U_t \cdot \gamma >$, that is the
measure of the change in $< V_j , x >$ in the step $t$. Then $Y_t$
is a Gaussian random variable with mean 0 and variance $\gamma
\cdot\twonormsq{ V_j }$.
\end{observ}
Note that the scaled
random variable $Y'_t = \frac{Y_t}{\gamma \twonorm{V^t_j}} $ is a
Gaussian with mean 0 and variance $\leq 1$. 
When $Y'_i$ correspond to standard Gaussian, then
\begin{lemma}
For any $\beta > 0$,
 $\Pr [ | Y'_1 + Y'_2 \ldots Y'_T | > \beta ] \leq 2 \exp ( - \beta^2 /2T
) $.
\label{chernoff-mart}
\end{lemma}

The behavior of random walks starting from an arbitrary initial position
and subsequently absorbed at 0 can be
obtained from the {\it gambler's ruin problem} where the underlying martingale
is the Brownian motion.
There exists a wealth of literature on Brownian motion \cite{feller1,Ross:2006},
but the specific form in which we invoke them for analyzing our algorithm
is stated below. A proof is presented in the appendix.

\begin{lemma}
Consider a random walk starting from position $a$ in an
interval of length $a +b$ with absorbing barriers at both end-points.
Then the expected number of steps for the walk to get absorbed (at any of the
ends) is $a \cdot b$.
Moreover, the probability of the random walk being absorbed at 0 is
$\frac{b}{a+b} - 1/k$
after $k\cdot a \cdot b $ steps for any $k > 1$.
\label{boost-prob-absorb}
\end{lemma}
{\it Remark} By choosing $b > k \cdot a$,
the probability
can be made arbitrarily close to 1 for $k \gg 1$ - conversely,
the probability of non-absorption at 0 is $O(\frac{1}{k})$ after $k^2 a^2$
steps.

\ignore{
The success of the rounding algorithm depends on the {\it race} between the
hitting of constraints $C^t_j$ and the unfixed 
variables of $V_j^t$. In particular, we would like to
hit the variable (hyper)-planes $x_i =1$ or $x_i = 0$ more often than the
hyperplanes $C^t_j$. 
 The sooner the ${(X_t )}_i$ hits the 
boundaries 0 or 1, the smaller is the displacement $\twonorm{ X_t }$ in 
${\mathbb R}^n$. 
Lovette and Meka \cite{LM:12} also make use of this observation
implicitly and in their case, the discrepancy hyperplanes are much further
compared to the variable hyperplanes. In our case, the situation is much 
tighter and much of our technical analysis revolves around that.  
If there are $r^2$ {\it unfixed} 
variables (a subspace with dimension $r^2$)
then the variance of the displacement in each step is proportional to 
the $L_2$ norm of $V_j$ which is $r$.
Since the
sum of the $r^2$ variables is bounded by 1 in any feasible LP solution, at 
least $r^2 /2$ variables have LP values bounded by $ x'_i \leq 
\frac{2}{r^2}$ which is the starting point of the random walk, ${(X_0 )}_i$. 
The 1-dimensional projection is also a brownian walk ${(X_t)}_i$
where each step is scaled by 
$\gamma$, and consequently, the $r^2 /2$ variables are within $\frac{2}
{\gamma \cdot r^2}$ from 0. Using a simpler argument based on number
of $\gamma$-length step sizes, one needs to move 
$\frac{\beta}{r \gamma}$ steps towards some constraint $C_j$ with 
slack $\beta$
versus $\frac{2}{\gamma \cdot r^2}$ steps to hit 0s of
variables associated with $C_j$. 
Clearly the 
latter is closer if $\beta > \frac{2}{r}$. This {\it bias} in favor of 
hitting the variables holds the key to the success and efficiency of 
our algorithm and we need to capitalize on this by choosing $T_p$
appropriately, so that the many variables become fixed at 0 before
$X_p$ hits the constraints (of the scaled polytope). This in turn slows
down the Brownian motion further, as the norm reduces and we keep repeating
this process. 
}

\newcommand{\qed}{\mbox{}\hspace*{\fill}\nolinebreak\mbox{$\rule{0.6em}{0.6em}$}} 
\newcommand{\ma}{{\mathcal A}}
\newcommand{\mf}{{\mathcal F}}
\newcommand{\hs}{\heartsuit}
\newcommand{\cs}{\clubsuit}
\newcommand{\noi}{\noindent}
\section{Brownian walk analysis}

We will divide the analysis into two components - 
First we will compute the rate at which the variables are absorbed at 0. 
Second, we compute the increments of all variables during each iteration that 
causes an inequality to be violated. 
This causes the right hand side of any inequality to 
behave as a martingale and we will refer to it as the {\it error}. 
 Note that the increase in {\it error} is related to the number of 
unabsorbed variables as they execute random walk. Once all variables are
absorbed, then the error doesn't change. For subsequent applications, we
will derive the bounds for starting position scaled by $S \geq 1$, i.e.,
from $\frac{x'}{S}$. Readers who are familiar with 
\cite{LM:12} may note that our analysis focuses on variables associated with
each constraint as opposed to the global number of variables.
\begin{lemma}

Let $x' \in \mathbb{R}^n$ be a feasible solution to $A\cdot x \leq B\cdot \vec{1}^m$, 
$A \in {\{0,1\}}^{m\times n}$, and 
$\overset{\sim}{x}=\frac{x'}{S}$, $S \geq 1$ be chosen as the starting point 
for brownian walk. Then,
 after $T_p=\frac{2^p. 2^p}{n^2 \gamma^2}$ steps, with probability 
$\geq 1-\frac{1}{m^{\Omega(1)}}$, the number of unfixed variables in $i^{th}$ 
constraint is $O(\frac{n B}{S \cdot 2^p}$) 
for $2^p \leq \frac{n B}{S  (\log (m)))}$\footnote{For $m$ polynomial in $n$
we need not distinguish between $\log m$ and $\log n$}.
\footnote{Applying a union bound over all constraints, we satisfy the conditions of the lemma for each of the phases and constraints with high probability}
\label{absorp_rate}
\end{lemma}
 First, we can
use Lemma \ref{boost-prob-absorb} to get a bound on the probability of 
absorption.
\begin{claim}
For $b \geq r a$,  after $O( \frac{ r^2 a^2}{\gamma^2} )$ Brownian motion 
steps, the probability of non-absorption at 0 is $\leq O(\frac{1}{r})$. 
\label{eq_var_bnd}
\end{claim}
The above claim is trivially true
for $r \leq 1$.
In Lemma \ref{boost-prob-absorb} , use $k = r $, $b = ra/\gamma$ 
that gives absorption probability $\leq O(1/r)$ after $r\cdot a/\gamma \cdot
a r/\gamma = \frac{ r^2 a^2 }{\gamma^2 }$ steps.   
\begin{proof}(of Lemma \ref{absorp_rate})
Choosing $T_p=\frac{2^p 2^p}{n^2 \gamma^2}=\frac{(\overset{\sim}{x_i})^2}{\gamma^2}(\frac{2^p}{n \overset{\sim}{x_i}})^2$,\\
we can apply Claim \ref{eq_var_bnd} with $r=\frac{2^p}{n \overset{\sim}{x_i}}$ to obtain that
the probability that $x_i$ is not absorbed at 0 is $\leq \frac{n \overset{\sim}{x_i}}{2^p}$.
Let $V_i^p\:x\leq B$ be the $i^{th}$ constraint after $p$ phases.
If $u^p_i$
(= $\twonormsq{ V_i^p }$) denotes the number of unabsorbed variables in 
constraint $i$ after $p$ phases then
$\E ( u^p_i ) 
\leq \sum_{j=1}^{n} \frac{n \overset{\sim}{x_{j}}}{2^p} A_{i,j} \leq 
\frac{n B}{S \cdot 2^p}$ since 
 $\sum_{j=1}^{n} A_{i,j}\overset{\sim }{x_{j}} \leq \frac{B}{S}$.
 
 Note that the variables are independently executing Brownian walks.
For $\frac{n B}{S \cdot 2^p} \geq  \log{(m)}$ or equivalently, 
$p \leq \log ({n})- \log (\log (m) )+\log (B)-\log (S)$,
 we can apply Chernoff bound to claim that $u_i^p$
is $O(\frac{n B}{S \cdot 2^p})$ with probability $\geq 1-\frac{1}{m }$ 
\end{proof}
\begin{cor}
Since $\E ( u^p_i ) =  \E ( \twonormsq{V_i^p} )$ for $V_i^p = {\{ 0, 1 \}}^n$,
we can bound $ \twonormsq{V_i^p} $ by $\frac{n B}{S \cdot 2^p}$ with high
probability for $p \leq \log ({n})- \log (\log (m \log(n)) )+\log (B)-\log (S)$.
\label{normbnd}
\end{cor}
\textbf{Remark}
(i) For $A_{i,j} \in [0,1]$, the above proof on $\E ( u_i^p )$ 
doesn't hold directly but the bound on $ \twonormsq{V_i^p} $ is still valid.
\\
(ii) If we proceed upto $p* = \log ({n})-\log (\log (m  ))+
\log (B)-\log (S)-\log (c)$, all constraints have
at most $O(\log m )$ 
unfixed coordinates.
 
\subsection{Error Bound}

From the previous result, 
for all $j , \twonormsq{V_j^p} \leq  O(\frac{nB}{S2^p}$) 
with high probability.
To bound the error consider a fixed constraint $C_j:V_j.x \leq B$,
we denote the error accumulated in the $p^{th}$ phase by $\delta_p$, so
that $\sum_{q=1}^p \delta_q = \beta_p$.

\begin{lemma}
For all constraints $C_j$, the error 
 $\delta_p \leq
\sqrt{\frac{2 B \log m } {S\cdot n }} 2^{\frac{p}{2}} .$\\ 
Thus the total error $\beta_p$ upto $\log n - \log (\log m )+ 
\log (B)- \log (S)$ is bounded by $c' \frac{B}{S}$ for some constant $c'$. 
\label{err_bnd} 
\end{lemma}
\begin{proof}
Consider $\Pr(< \gamma (\sum_{i=T_{p-1}+1}^{T_{p}}U_i ),V_j>\vert \geq \delta_p )$\\ 
=$\Pr(\vert < \sum_{i=T_{p-1}+1} ^{T_p} U_i, \frac{V_j}{\vert \vert V_j \vert \vert}> \vert \geq \frac{\delta_p}{\gamma \vert \vert V_j \vert \vert}).$\\

Now $<U_i,\frac{V_j}{\vert \vert V_j \vert \vert}> 
\sim \mathcal{N}(0, \sigma^2)$ where $\sigma^2 \leq 1$. From Lemma 
\ref{chernoff-mart}, it follows that
the above is bounded by $\exp (-\frac{(\delta_p)^2}
{\gamma^2 \vert \vert V_j \vert \vert ^2 2 (T_p-T_{p-1})})$.
It will hold simultaneously for all the $m$ constraints in the $p^{th}$ 
phase, if $\delta_p$ satisfies :-

$\frac{(\delta_p)^2}{\gamma^2 \vert \vert V_j \vert \vert ^2 
2 (T_p-T_{p-1})} \geq \Omega( \log m )$\\
i.e. if $\delta_p \geq \gamma \vert \vert V_j \vert \vert \sqrt{2 \cdot
\log  m .(T_p-T_{p-1})}$\\
From Corollary \ref{normbnd} $\twonormsq{V^{p-1}_j } \leq \frac{nB}{S\cdot 2^{p-1}}$ and $T_p=\frac{2^p 2^p}{n^2 \gamma^2}$,\\ 
it suffices to choose 
 \[ \delta_p \geq \gamma \sqrt{\frac{nB}{S\cdot 2^{p-1}}}\cdot 
\sqrt{2 \log m }\cdot \sqrt{\frac{2^p\cdot 2^p}{n^2 \gamma^2}}
 =2 \sqrt{\frac{ B \log m  2^p}{S\cdot n}}
\]
The above bound for $\delta_p$ holds with high probability when $p \leq \log n -
\log (\log m )+ \log (B)- \log (S)$ since only in this situation 
we can bound effective value of $V_j$.\\
Thus the total error upto $\log n - \log (\log m )+ \log (B)- \log (S)$ is 
bounded by 
\begin{eqnarray}
\sum_{p=1}^{\log n -\log (\log m )+\log (B)-\log (S)} 2\sqrt{\frac{ B\log m }
{S\cdot n}} 2^{\frac{p}{2}} &
 \leq & 2\sqrt{\frac{ B\log m }{S\cdot n}} \cdot O(2^{\frac{\log n - 
\log (\log m )+\ log B - \log (S)}{2}}) \\
 & = & 2\sqrt{\frac{ B \log m }{S\cdot n}} \cdot O(\sqrt{\frac{nB}
{S\cdot \log m }})=O(\frac{B}{S})
\end{eqnarray}
Therefore, the total error in every constraint is bounded by
 $<V_j,x^p>=<V_j,\overset{\sim}{x}+\gamma\sum_{i=1} ^{T_p} U_i>=
O(\frac{B}{S})$.
\end{proof}
Thus the solution obtained satisfies the constraints 
$A \cdot x \leq \frac{B}{S}$.

 The above method can be extended to obtain an alternate proof the 
$O(\log m/\log\log m)$ error bound of Raghavan-Thompson that we
have ommitted from this version. 

Based on the above results, we summarize as follows.
\begin{cor}
In the first stage of Algorithm {\bf Iterative Randomized Rounding}, we 
run the algorithm for
$T = \frac{B^2}{S^2 \log^2 m \gamma^2}$ steps. Then with high 
probability\\ 
(i) all constraints have $ \leq \log m$ unfixed variables and
\\
(ii) the total error in any constraint is bounded by $c'B/S$ for some
constant $c'$. 
\label{browniansteps}
\end{cor}

Using {\it multiple copies} of a variable, results similar
to this section have been used before that are {\it ad hoc} to specific 
applications \cite{CC:09,AS:96}.
However, it requires two stages of rounding - once choosing exactly one
copy followed by a phase of independent rounding. In comparison, our
technique and analysis are more general. 
 

\section{Applications using LLL}

\label{short}

At the end of brownian walk(that is after $\frac{B^2}{S^2\log^2m \gamma^2}$ steps), we have atmost $\log(m)$ unconverged variables 
occurring(with non zero coefficients) in each equation.
Now we need to bound the error in independantly rounding the unconverged 
Now we need to bound the error in independantly rounding the unconverged 
variables.
As per the notations setup in section \ref{sec:algo}, $V_i^{p*}$ represents the 
coefficient vector of $i^{th}$ constraint after p* phases.
\\

Define $\overset{\wedge}{A}$ as a matrix having rows as $V_i^{p*}$.

Now if the unconverged variables are rounded from $\overset{\sim}{x}$ to $\overset{\wedge}{x}$,
the change in the value of RHS is 
$A\cdot(\overset{\wedge}{x}-\overset{\sim}{x})=\overset{\wedge}{A} \cdot (\overset{\wedge}{x}-\overset{\sim}{x})$.
(Since only unconverged variables change).

Hence to calculate additional error due to independant rounding we only need to consider the "`sparsified"' matrix $\overset{\wedge}{A}$.
\\
Before we prove the main results, we consider the simpler 
case of columns with bounded number of ones -
Suppose the matrix $A^{m\times n}$ has no more than
$\rho$ 1's in any column. Let $OPT$ be the optimal fractional 
objective value for the weighted objective function $\sum_i c_i \cdot x_i$. 
Consider a fixed constraint $C_r$, that contains $\log m$ 1's after the Brownian
motion and let $j_1 (r) , j_2 (r) \ldots $ denote the (at most) $\log n$
columns that contain 1.
We say that a constraint $C_y$ is {\it dependent } on $C_r$ if they
share at least one column where the value is 1.
So, the dependency of any single constraint can be bound by $\rho\log m$
(Since there are $\log m$ 1's per row in the sparsified matrix $\overset{\wedge}{A}$).
If we use the independent rounding to round the 
fractional solution $x'$, the probability that the value of a constraint
exceeds $t > 1$ is bounded $\frac{1}{2^t}$ from Chernoff
bounds \footnote{This
is a slightly weaker version to keep the expression simple}. Let
$E_i$ denote the event that $C_i$ exceeds $t$ when we use randomized
rounding.  We are interested to know the
probability of the event $\bigcap_{1 \leq i \leq m} \bar{E_i}$ 
since this implies the event that all the inequalities are less than $t$.
This is tailor-made for Lovasz Local Lemma (LLL).
We also want to guarantee a large
value of the objective function. For example, setting all variables equal to
zero would guarantee feasibility but also return an objective function value 0.
Thus we define an additional event $A_{m+1}$ corresponding to the objective
function value less than $(1- \epsilon ) \cdot OPT$ for some suitable $1 >
\epsilon > 0$.
Since $A_{m+1}$ is a function of all the variables, it has dependencies with
all other $A_i \ \ i = 1 \ldots m$, therefore we have to use the generalized
version of LLL in this case. 
\begin{theorem}[Lovasz Local Lemma \cite{EL:75}]
Let $A_i , 1 \leq i \leq N$ be events such that $\Pr [ A_i ] = p$ and each
event is dependent on at most $d$ other events. Then if $ep(d+1) < 1$, then\\
$ \Pr ( \bar{A_1} \cap \bar{ A_2 } \ldots \bar{ A_N } ) > 0 $.
\\
Alternately, in a more general (asymmetric) case, where the dependencies are
described by a graph $(\{ 1, 2 \ldots N \}, E)$
where an edge between $i, j$ denotes dependency between
$A_i , A_j$ and $y_i$ are real numbers such that $\Pr ( A_i ) \leq
y_i \cdot \prod_{(i,j) \in E} ( 1 - y_j )$ then
$ \Pr \left( \bigcap_{i=1}^{N} \bar{ A_i } \right) \geq \prod_{i=1}^{N}
( 1 - y_j )$.\\
Moreover, such an event can be computed in randomized polynomial time using
an algorithm of Moser and Tardos \cite{MT:10}.
\label{lll}
\end{theorem}

If we choose $t$ such that $e \cdot 2^{-t} \cdot (d+1) \leq 1$ or
equivalently $t = \log d$, then we can apply
the previous theorem to obtain a rounding that satisfies an error bound of
$O( \log \rho + \log\log m )$. 
The error $t$ can be improved to $O(\frac{\log d}{\log\log d})$
by using a tighter version of the Chernoff bound (Equation \ref{chernoff}).


We define $A_{m+1}$
as the event where the objective value is less than $(1- \epsilon ) OPT$. From
Chernoff-Hoeffding bounds, we know that $\Pr ( A_{m+1} ) \leq 
\exp ( - \epsilon^2  OPT/2 )$.
We define $y_i$ for $i = 1, 2 \ldots m$ as before, corresponding
to the probability of exceeding $t = \log d/\log\log d$.
\\ By choosing
$y_i = 1/(\alpha d) \ i \leq m$ and $y_{m+1} = \frac{1}{2}$, 
for some suitable scaling factor $\alpha \geq e$, we
must satisfy the following inequalities
\begin{eqnarray}
\Pr (A_i ) & \leq & 1/(\alpha d) {( 1 - 1/(\alpha d) )}^{d} \cdot \frac{1}{2} \ \  i = 1, 2 \ldots m \\
\Pr ( A_{m+1} ) & \leq & \frac{1}{2} 
{( 1 - 1/(\alpha d) )}^{m} \leq
\frac{1}{2}\cdot \exp (- m/(\alpha d))
\label{eqnobj}
\end{eqnarray}
The first condition is easily satisfied when $\Pr ( A_i ) = \frac{1}{\Omega 
(\alpha d)}$.\\
To satisfy the second inequality, we can choose $\alpha$
so that  $OPT \geq \Omega ( \frac{m}{\alpha d})$. 
This implies that $\exp (- m/(d\alpha ) \geq
\exp (- \epsilon^2 OPT/2))$,  so condition (\ref{eqnobj}) is
satisfied for LLL to be applicable. \\


\begin{figure}
\fbox{\parbox{6.0in}{

{\footnotesize 
\begin{center}
Algorithm {\bf LLL based Iterative Randomized Rounding }
\end{center}
Input : $x'_i \ \ 1 \leq i \leq n, t$ (error parameter)
\\
Output : $\hat{x_i} \in \{0, 1\}$ \\

Do independent rounding on all the variables having values $x'_i$ to
$\hat{x_i}$.
\\
Compute the value of each constraint $C_i$ as $< V_i , \hat{x} >$

{\bf While} any inequality exceeds $t$ {\bf or} objective value is $< OPT/2$
\begin{quote}
\begin{enumerate}
\item Pick an arbitrary constraint $C_j$ that exceeds $t$
and perform independent rounding on all the variables in $V_j$.
\item Update the value of the constraints whose variables have changed
\end{enumerate}
\end{quote}
Return the rounded vector $\hat{x}$.
}}
}
\caption{An iterative randomized rounding algorithm based on Moser-Tardos}
\label{algo2}
\end{figure}

We summarize our discussion as follows
\begin{lemma}

For $A \in {\{0,1 \}}^{m \times n}$ with a maximum of 
$\rho$ 1's in each column, we can round 
the optimum solution $x^{*}$ of the linear program
$\max_{x} \sum_i c_i x_i \ \ s.t. A x \leq 1 \ \ , 0 \leq 
x_i \leq 1 $ 
to $\hat{x} \in {\{ 0, 1 \}}^n$ such that 
\[
{|| A \hat{x}||}_\infty \leq max(\left( \frac{\log \rho +  \log\log m}{\log\log 
(\rho \log m)}\right),\frac{\log(\frac{m}{OPT})}{\log(\log(\frac{m}{OPT}))}) \ \ \text{ and } \sum_i c_i \hat{x}_i \geq 
(1 - \epsilon ) OPT   
\label{bnd_col}
\]
\end{lemma}
\begin{proof} 
The error bound follows from the previous discussion by setting $y_i = \Omega (
\frac{1}{\alpha d})$ and using the stronger form of Chernoff bound in
Equation \ref{chernoff}. 
Note that Equation \ref{eqnobj} can be satisfied by ensuring
$OPT \geq \frac{m}{\alpha d}$ by choosing an appropriately large $\alpha$.\\
The above is satisfied for $\alpha \geq \frac{m}{OPT.d}$ i.e. if  $\alpha.d=\frac{m}{OPT}$.\\
In this case the discrepancy is $\frac{\log (\frac{m}{OPT})}{\log(\log(\frac{m}{OPT}))}$\\
Hence the result follows.
\end{proof}

The objective function still follows the martingale property since the
variables starting the random walk at $x'_i $ have probability
$ x'_i $ of being absorbed at 1 which is identical to  the independent
rounding that we use in the second phase
for Moser-Tardos algorithm.
Although, the random walk is short-cut by a single step independent
rounding, the distribution for absorption at 0/1 remains unchanged. To
see this, consider the {\it last} time a variable $x_i$ is rounded by the
Moser-Tardos algorithm - the probability is $x'_i$ since it is done
independently every time.

\ignore{
The rest follow from LLL and the algorithm of Moser-Tardos. The expected
running time for the second phase is $O( \frac{ m^2 }{n} polylog (n))$ in
our case.\\
That the error bound may fail if either the
Brownian walk and the Moser-Tardos algorithm fail
with inverse polynomial probability because of the inherent randomization.
This probability can be boosted by repetition.
\end{proof}
As an example, consider $\rho = \log m$, $OPT \geq \frac{m}{\rho 
\log^{c+1} m }$ then for $\alpha = \log^{c} m$,
the error is $O(\log\log m/\log\log\log m )$.
\\
{\it Remark} A direct application of LLL (without running the Brownian
walk) would increase the dependence to $O(\rho \cdot n)$ resulting in an
error bound of $O( \frac{\log \rho + \log m}{\log\log m })$ that is not 
better than the RT bound. 
}
\subsection{Random matrices}
\label{sec5}

Let ${\cal A}^{m\times n}_k$ denote the family of 
$m \times n$ 0-1 matrix with exactly $k$ 1s in each of
the $m$ rows chosen uniformly at random. 
Clearly $x'_i = 1/k$ is a feasible solution with objective
function value $n/k$. After rounding $x'$ to $\hat{x}$, 
we want to achieve an objective value
$\Omega ( n/k )$ 1s in the rounded vector. 
In addition, $A \cdot x' \leq t\cdot e^m$ for error guarantee
$t$. 

To compute the dependency $d$ for $C_r$, we observe that 
another constraint $C_i$ will contain a 1 in $j_1 (r)$ with 
probability $\frac{k}{n}$, 
i.e., if it had one of the $k$ randomly chosen 1's in that column, which
is $\frac{k}{n}$. Since all the rows were chosen independently, the expected
number of rows among $m$ rows that have a 1 in column $j_1 (r)$ can
be bounded by $\frac{m k}{n}$ and by $O(\max\{ \frac{m k }{n} , \log m \} )$ 
with probability greater than $1 - 1/m$. Since this holds for all the positions,
$j_i (r)$, by the union bound, and 
using $\max \{ a , b \} \leq (a+b)$, we can claim the following
\begin{claim}
The total number of constraints that are 
correlated to $C_r$ can be
bound by $O(\frac{m k \log m}{n} +  \log^2 m )$ with high 
probability.
\label{deg_bnd}
\end{claim}
In our context, for $k = \log m$, $d = O( \frac{m\log^2 m}{n} + 
\log m ^2)$. 
We can verify this by 
exhaustively computing the dependence - if it exceeds this then our algorithm
is deemed to have failed which is bounded by inverse polynomial probability.
If we choose $t$ such that $e \cdot 2^{-t} \cdot (d+1) \leq 1$ or
equivalently $t = \Omega ( \log (\frac{m\log^2 m}{n}+ \log^2 m  ) ) $, 
then we can apply
the previous theorem to obtain a rounding that satisfies an error bound of
$O( \log (\frac{m\log^2 m}{n}+ \log^2 m ) )$. \\
For $m$ bounded by 
$n \log^{O(1)} m$ this is
$O(\log\log m)$ which is substantially better than the Raghavan-Thompson
bound.  \\However for $m \geq n^{1 + \epsilon}$ for any constant $\epsilon > 0$,
it is no better.
\begin{proof} (Completing the Proof of Theorem \ref{mainthm0-1})
Now we formally define our bad events $E_i$ here.\\
$\forall 1\leq i \leq m$, define $E_i\equiv(A_i^T x > \delta )$.(where we choose error to be $\delta$).
Define $E_{m+1}\equiv(c^T.x <(1-\epsilon)OPT)$.
We have $Pr(E_i)\leq \frac{e^{\delta}}{{1+\delta}^{1+\delta}}$.
and $Pr(E_{m+1})<e^{\frac{-\epsilon^2.OPT}{2}}$.
We choose $y_i=\frac{1}{\alpha.d}$,for some $\alpha\geq1$.
and choose $y_{m+1}=\frac{1}{2}$.
To apply LLL we need,
$Pr(E_i)<\frac{1}{\alpha.d}(1-\frac{1}{\alpha.d})^d.(1-\frac{1}{2})$
and $Pr(E_{m+1})<\frac{1}{2}(1-\frac{1}{\alpha.d})^m$.
To satisfy these equations we get error,$\delta=\frac{\log(\alpha.d)}{\log(\log(\alpha.d))}$.\\
Also we get $e^{\frac{-\epsilon^2.OPT}{2}}<e^{\frac{-m}{\alpha.d}}$.
The second equation is satisfied for $OPT>\frac{m}{\alpha.d}$ or $\alpha.d >\frac{m}{OPT}$.
Combining the 2 results we have error $\delta$ in constraints as $max\left\{\frac{\log(d)}{\log(\log(d))},\frac{\log(\frac{m}{OPT})}{\log(\log(\frac{m}{OPT}))}\right\}$\\which is $max\left\{\left( \frac{\log (\frac{mk\log(m)}{n}) +  \log\log m}{\log\log 
(\frac{mk \log(m)}{n} \log m)}\right),\frac{\log(\frac{m}{OPT})}{\log(\log(\frac{m}{OPT}))}\right\}$.
When $c_i \geq \frac{1}{p}$ we have $OPT \geq \frac{n}{k.p}$.
Substituting the value of OPT in the equation and using the fact that $\frac{mk}{n}=O(\log(m))$, we prove Theorem \ref{mainthm0-1}. 
\ignore{
\color{red}

 Then, using 
$d \leq \frac{mk\log m}{n}$ (Claim \ref{deg_bnd}), $\exp (- m/(d\alpha ) \leq
\exp (- n/(\alpha k\log n))$ so the condition \ref{eqnobj} is 
satisfied for LLL to be applicable. In general, Equation \ref{eqnobj} can be
satisfied when the objective function value $OPT$ 
exceeds $\frac{n}{\alpha k\log n}$ for choosing a suitable $\alpha \geq e$. 

\subsection{Putting it together}
For a random $m \times n$ matrix $A \in {\cal A}^{m \times n}_k$ 
for $k \geq \log n$ 
in $A^k$, we first apply the Brownian
walk based algorithm to sparsify it to $\log m$ unfixed variables
per constraint. 

For weighted objective function, observe that Equation \ref{eqnobj} 
can be satisfied for value of the objective function at least $\Omega 
(n/(k\log n)$. This can be ensured by the condition $c_i \geq 1/\log n$.
More generally, for $c_i \geq 1/(\alpha \log n)$, 
we can use a scale factor $\alpha \geq 1$ to obtain
an error bound of $O(\log (d \alpha) /\log\log (d \alpha))$ with
objective value $\Omega ( n/(k\alpha \log n))$.\\ 
The rest follow from LLL and the algorithm of Moser-Tardos.} The expected 
running time for the second phase is $O( \frac{ m^2 }{n} polylog (n))$ in
our case, due to application of Moser Tardos.\\ 
Here the error bound may fail for the additional reason that
the dependence may exceed $\frac{m k polylog (n)}{n}$ after the Brownian 
motion phase. This can happen with probability $\frac{1}{m^{\Omega (1) }}$. 
This part is related the input distribution of ${\cal A}^{m\times n}_k$ and
repeating the algorithm may not work.
\end{proof}

{\it Remark} A direct application of LLL (without running the Brownian 
motion) would increase the dependence to $O(\frac{m k^2 \log n}{n} +\log m 
\log n )$. In order to maintain the same asymptotic error bound, the number
of rows in the matrix, $m$, could be significantly less. For example, if
$m,k = n^{1/2}$, then the difference would be $O(\log n/\log\log n)$
versus $O(\log\log n/\log\log\log n )$.  

The proof of Theorem \ref{mainthm3} is along similar lines after appropriate 
scaling and is given in the appendix.

\ignore{
To obtain the approximation factor of 
Theorem \ref{mainthm2} we use the previous scaling technique with
$d = \alpha \max\{ \frac{e mk u}{n} , u\log m\}$, where $u$ is the
number of unabsorbed variables in any constraint and $\alpha \geq 1$ is an
appropriate scaling factor. Let $B' = B+1$, then 
\[ \beta = e \alpha^{1/B'} \cdot \max\{ {(\frac{e mk u}{n})}^{1/B'} ,
 {(u\log m)}^{1/B'} \} = e \cdot {(\alpha u)}^{1/B'} 
\cdot \max\{ {(\frac{e mk }{n})}^{1/B'} ,
 {(\log m)}^{1/B'} \} . \]

For weighted objective function with $c_i \geq 1/p$, we can choose
the value of the scaling factor $\alpha$ more carefully so that $p = 
\alpha^{1 - 1/B'} \cdot {(\log n)}^{1 -2/B'} $, and $mk/n \leq \log m$,
using $S = {(\alpha \log^2 n )}^{1/B'}$ gives an objective value
of at least $ nB'/(p Sk)  \geq n/(k \alpha \log n) $ 
which satisfies condition \ref{eqnobj}. Note that $S = {( p \log n)}^{1/(B'-1)} 
$ from the above parameter settings.\\  
By subsituting back $B+1$ for $B'$ in the previous calculations,
$S = {( p \log n)}^{1/B}$. So for $B = 1$, the approximation is
$O(p\log n)$.
This completes the proof of Theorem \ref{mainthm2}. \\[0.1in]
}
\color{black}

\section{Applications of our rounding results}
\label{sec6}

In this section we briefly sketch the applications of our rounding theorems.
Although these do not significantly improve prior results,
our parameterization could simplify further applications.
\ignore{For some of the applications, we need a weighted version of Theorem
\ref{mainthm2} that we state without proof - which follows from our
earlier applications of LLL.
\begin{lemma}
For $A \in {\{0,1 \}}^{m \times n}$ with a maximum of
$\rho$ 1's in each column, we can round
the optimum solution $x^{*}$ of the linear program
$\max_{x} \sum_i c_i x_i \ \ s.t. A x \leq 1 \ \ , 0 \leq
x_i \leq 1 $ and
 $1 = \max_i c_i \geq c_i \geq 1/\alpha \mbox{ for some } \alpha \geq e $
to $\hat{x} \in {\{ 0, 1 \}}^n$ such that
\[
{|| A \hat{x}||}_\infty \leq \left( \frac{\log \rho + \log \alpha + \log\log n}{\log\log
(\rho \alpha \log n)}\right) \ \ \text{ and } \sum_i c_i \hat{x}_i \geq
(1 - \epsilon ) OPT \text{ for } OPT \geq \Omega ( m/(\alpha d ))
\label{bnd_col}
\]
\end{lemma}
}

\subsection{Application to Switching circuits}

Consider an $n$-input butterfly network (with $n\log n$ total nodes) 
with $( s_i , t_i ) \ \ i \leq n\log n$ source destination
pairs where each input/output node has $\log n$ sources and $\log n$ 
destinations. 
For any arbitrary instance of routing, we can do a two phase
routing with a random intermediate destination, say, we choose a random
intermediate destination $r_i$ for the source-destination pair $(s_i , t_i )$. 
We want to route a maximum number of pairs subject to some edge capacity
constraints. 

If $n = 2^L$, then the the expected
congestion $k$ of an edge for a random permutation is $\log n \cdot 
\frac{ 2^{\ell} \times
2^{L - \ell} }{2^L} = \log n$ and moreover it can be 
bounded by $c\log n$ with high probability for some constant $c$. 
So there exists a fractional solution with flow value $\frac{1}{\log n}$
for each of the $n\log n$ paths with objective function value $n$. 

Let $A$ be an $m \times t$ matrix where $m$ is the number of edges
and $t = n\log n$ is the number of paths. The edges are denoted 
by $e_1 , e_2 \ldots e_m$
and the paths are denoted by $\Pi_1 , \Pi_2 \ldots \Pi_t$. Then 
$A_{i,j} =1$ iff 
the flow $\Pi_j$ passes through edge $e_i$. The number of edges in a path
is bounded by $L = \log n$. The value of the flow through 
path $\Pi_j$ (
denoted by $f_j$) is the amount of flow and let $\bar{f}$ be the
vector denoting all the flows. Then $A \cdot f \leq \bar{c}$ where $\bar{c} =
( c_1 , c_2 \ldots c_m )$ is the vector corresponding to the congestion
in edges $( e_1 , e_2 \ldots e_m )$.

Since the congestion is bounded by $c \log n $, $A \in {\cal A}^{m \times 
n\log n}_{c\log n}$ and there exists a fractional solution 
$\bar{f} = ( 1/c\log n , 1/c\log n \ldots 1/c\log n$ with objective value
$\Omega (n)$. Using $\rho = O(\log n)$ in 
Lemma \ref{bnd_col}, we can round it to
a 0-1 solution that yields the following result. 
\begin{theorem}
In an $n$ input BBN butterfly network having $2(n \log n)$ edges, 
we can route the $\Omega (n)$ source-sink pairs
with congestion $O(\log\log n/\log\log\log n )$. From the capacity 
constraint, it follows that no node contains more than $O(\log\log n)$
sources or destination. 
\end{theorem}
Note that we can also handle weighted objective functions. 

The above result matches the previous results of \cite{MS:99,CMHMRSSV:98} 
which have the advantage of being online. Using $B = 
\log\log n/\log\log\log n$, in Lemma \ref{colbndapprox} we can obtain an
optimal bound. Note that, in this case we can match asymptotically 
the optimal fractional flow of $n\log\log n/\log\log\log n)$ with 
$x_i = \frac{\log\log n}{c \log\log\log n\cdot \log n}$.
\begin{theorem}
In an $n$ input multi-butterfly network having $n\log n$ edges, 
we can route $\Omega (\frac{n\log\log n}{\log\log\log n})$
flows with congestion $O(\log\log n/\log\log\log n )$ 
where each node can be the source or sink of at most 
\\
$O(\log\log n/\log\log\log n )$ flows.
\end{theorem}
The above result marginally extends a similar result of Maggs and Sitaraman
\cite{MS:99} where they can route $(n/\log^{1/B} n)$ pairs in the online
case. It remains an open problem to find a fast online implementation of
our rounding algorithms. 
\subsection{Maximum independent set of rectangles}

Consider a $\sqrt{n} \times \sqrt{n}$ grid and a set of axis-parallel
rectangles that are aligned with the grid points,i.e., the upper-left
and the lower-bottom corners are incident on the grid\footnote{We can
assume that all the grid-points that are flush with the sides are
in the interior by slightly enlarging the rectangle}. Each rectangle
contains $A$ grid points for some $A \leq n$ but they can have any
aspect ratio. Trivially, the different types of such rectangles can be
at most $A$ where $A = a \cdot b$ for $a = 1, 2 \ldots A$. Two such
rectangles $r_1$ and $r_2$ are {\it overlapping} if $r_i \cap r_2$
contains one or more grid points. 

From the $A$ possible rectangles whose upper-right corners are anchored
at a specific grid point $p$, the input consists of one such rectangle.
To avoid messy calculations, we assume that the grid is actually embedded on
a torus. Given a set $S$ of $n$ such rectangles, our goal is to select a
large non-overlapping set of rectangles. Although this is a restricted
version of the general MISR problem it still captures many applications. 

For any given grid point $p$, let $S(p)$ denote the set of rectangles
containing $p$. This can be bounded by $\sum_{a = 1}^{A} a \times A/a = 
O( A^2 )$. After solving the relevant packing linear program, the
$n \times n$ point-rectangle matrix ${\cal A}$ contains 1 in the $(i,j)$
position if the $i$-th grid point is incident on the $j$-th rectangle.
From our previous observation, no row contains more than $A^2$ 1's. 

By setting $k = A^2$, we can apply our rounding results to obtain a wide
range of trade-offs for this problem. For example, if $A$ is polylog($n$), 
then using $\rho = A$ in Lemma \ref{bnd_col}, 
we can choose $\Omega (OPT)$ rectangles with no more than 
$O(\log\log n/ \log\log\log n )$ overlapping rectangles on any clique, 
where $OPT$ is the fractional
optimal solution of the corresponding packing problem. 

Further, using a result of \cite{ENO:04}, we
can obtain $\Omega ( \frac{OPT \log\log\log^2 n}{\log\log^2 n })$ 
non-overlapping rectangles as well as a $\frac{\log\log^2 n}
{\log\log\log^2 n}$ approximate solution in the weighted case.
\subsection{$b$-matching in random hypergraphs}

In a hypergraph $H$ on $n$ vertices has hyper-edges defined by subsets of 
vertices.
We wish to choose the maximum set of hyper-edges such that no more than
$b$ hyperedges are incident on any vertex. An integer linear program
can be written easily for this problem with $n$ constraints and $m$
variables corresponding to each of the $m$ edges. Let $A$ be an $n \times m$
matrix where $A_{i,j} = 1$ if vertex $i$ is incident on the $j$-th hyperedge.
Let $x_i = \{ 0,1 \}$ depending on if the $i$-th hyperedge is selected in the
$b$-matching.
\[ \textbf{Maximize} \sum_i x_i \ \ s.t. \ A x \leq b \cdot e^m \ \ x_i \in
\{ 0, 1 \} \] 
The relaxed LP can be solved and the solution may be rounded. The weighted
version can be formulated similarly by associating weights $w_i$ for $x_i$.

If the $m$ edges of $H$ are random 
subsets of vertices, then we can represent the fractional 
solution as an $n \times
m$ random matrix with each entry of the matrix set to 1 with probability
$k/n$. That is, each hyperedge has $k$ vertices and chooses $k$ vertices 
randomly. So the expected number of edges incident on a vertex is 
$\frac{mk}{n}$. There exists a feasible fractional solution using
${x_i} \ \ 1 \leq i \leq m = \frac{n}{mk}$ with objective value $OPT \geq n/k$.

For $b = \Omega ( \log\log n/\log\log\log n )$, we can can obtain a 
$b$ matching of
size $\Omega (OPT)$ for {\it most} hypergraphs which matches the best
possible size given by the fractional optimum. In general, we obtain 
an approximation $k^{1/b}$ for $b \geq 2$ using the result of Theorem
\ref{mainthm3} including the weighted version. 

It is known that $k$-uniform $b$-matching problem cannot be approximated
better than $\frac{k}{b\log k}$ for $b \leq k/\log k$ unless 
$P = NP$ \cite{OFS:11,HSS:06}. So our
result shows that the bound can be much better for many $k$-uniform
hypergraphs, for example $k = \log n$ and $b \geq 2$. 
 A closely related result for bounded column case was observed by
Srinivasan \cite{AS:96}.

\begin{footnotesize}
{\bf Acknowledgement} The author would like thank Antar Bandhopadhyay for
pointing out the generality of the Optional Stopping theorem and Deeparnab
Chakrabarty for pointing out the result in 
\cite{CGKT:07}.
A lower-bound construction was given by Nisheeth Vishnoi and Mohit Singh
that refuted an earlier stronger claim by the author. The authors also 
acknowledges Nikhil Bansal for some useful observations with regards to
formalizing the main theorem.
\bibliographystyle{alpha}
\bibliography{round}
\end{footnotesize}
\appendix
\section{Appendix}
{\bf Proof of Lemma \ref{boost-prob-absorb}}
\begin{proof}
The proof of the gambler's ruin problem actually uses a {\it quadratic
martingale}\footnote{The proof that it is a martingale can be found in
standard books on stochastic process}  $B^2 (t) - t$ where $B(t)$ is the 
Brownian motion random variable.
Using the optional stopping theorem on this, we obtain
\[ \E [ B^2 (T) - T ] = \E[ B^2 (0) - 0 ] = a^2 \]
Since $p = \frac{b}{a+b}$ is the probability that it is at 0, we obtain
$\E[T] = a \cdot b$.

Now, consider the events leading to the failure of being absorbed at 0 -
these correspond to the absorption of the random walk at either ends or
the non-absorption at either ends. 
\[ \Pr [ Failure ] = \Pr [ Failure \cap Absorption ] + \Pr [ Failure \cap
Non-absorption ]\]
\[  = \Pr [ Failure | Absorption ] \cdot \Pr [ Absorption] 
+ \Pr[ Failure | Non-absorption ] \cdot \Pr [Non-absorption]  \]
\[
\leq \Pr [ Failure | Absorption ] + \Pr [Non-absorption] \] 
From the preceding argument ,the probability that the random walk is not 
absorbed after $k a\cdot b$ steps 
is less than $1/k$ using Markov's inequality. 
From optional stopping criteria, the probability of
absorption at 0 is $\frac{b}{a+b}$. The probability that the random walk is
not absorbed at 0 after $k a  b $ steps is bounded by 
$ 1 - \frac{b}{a+b} 
+ 1/k$. Consequently, the probability of being absorbed at 0 is 
at least $\frac{b}{a+b} - 1/k$. 
\end{proof}
\subsection{A Lower Bound on error}
\begin{lemma}
For $A \in {\cal A}^{m\times n}_{\log n}$ such that $A\cdot \bar{x} \leq e^m$,
we cannot simultaneously obtain $\sum_i \hat{x}_i = \Omega ( n/\log n )$ and
${|| A \hat{x}||}_\infty  \leq o(\frac{\log\log n}{\log\log\log n})$
for $\hat{x}_i \in \{ 0, 1 \}$.
\label{lbnd}
\end{lemma}

\begin{proof}

As mentioned before, it is easy to see that $x_i = 1/k$ is a solution 
for $\bar{x}$ with objective
function value $n/k$. After rounding $\bar{x}$ we want to have at least 
$\Omega ( n/k )$ 1s in the rounded vector, say $\bar{y}$.
We must have $A \cdot \bar{y} \leq t\cdot e$ for some rounding guarantee
$t$. Let $b$ be a fixed vector with $n/k$ 1s. 

Let $r$ be a row vector with $k$ 1s in random location - this corresponds to
a row of $A$. The probability that $< r , \bar{y} > \ \ \geq t$ is given by
 \[  \frac{{ (n/k) \choose t} 
\cdot { (n-n/k) \choose (k-t) }}{{n \choose k }} \]
Using $ {(\frac{n}{k})}^k \leq {n \choose k} \leq {(\frac{ne}{k})}^k $, the 
above expression is
\begin{eqnarray}
 \geq & \frac{ {( \frac{n}{tk})}^t \cdot {(\frac{n -n/k}{k-t})}^{k-t}}{
{(\frac{ne}{k})}^k} \\
 = & \frac{ \frac{1}{t^t} \cdot {( \frac{k-1}{k-t})}^{k-t}} {
 e^k } \\ 
 = & \frac{ k^{k-t}}{ {( k-t )}^{k-t} \cdot e^k \cdot t^t }\\
 = & \frac{1}{ {( 1 - t/k )}^{k-t} \cdot e^k \cdot t^t } \\
 \geq & \frac{1}{{( 1 - t/k )}^{k/2} \cdot e^k \cdot t^t } \text{ for } k \gg t\\
 = & \frac{1}{ e^{ k - t/2} \cdot t^t }
\end{eqnarray}
So the probability that for $m$ independently chosen rows, the probability that
$< r , \bar{y} > \ \ < t$ is 
\begin{eqnarray}
  p(m,k,t) &  \leq & {\left( 1 - \frac{1}{ e^{ k - t/2} \cdot t^t } \right)}^m 
\leq  {\left( 1 - \frac{1}{ e^{ k } \cdot t^t } \right)}^m =
{\left( 1 - \frac{1}{ e^{ k } \cdot t^t } \right)}^{e^k t^t \frac{m}{e^k t^t}}
\\
 & \leq & {\left( \frac{1}{e} \right)}^{\frac{m}{e^k t^t}} 
\end{eqnarray} 
Since there are no more than ${n \choose
(n/\log n)}$ choice of columns with $n/\log n$ 1s, the probability that all
the dot products are less than $t$ is less than ${n \choose (n/\log n)} \cdot
p(m,k,t)  \leq p(m,k,t) \cdot {(\frac {n e}{n/\log n})}^{n/\log n} 
 \leq e^n \cdot p(m,k,t)$. If $e^n \cdot p(m,k,t) < 1$ then there must exist
matrices in ${\cal A}^{m \times n}_k$ that do not satisfy the error bound $t$.
So,
\[    1  > {\left( \frac{1}{e} \right)}^{\frac{m}{e^k t^t}} \cdot e^{n} 
    \Rightarrow {( e ) }^{\frac{m}{e^k t^t}}  >  e^n \] 
 which is equivalent to the condition 
\begin{equation}
 \frac{m}{n}  > e^k t^t \text{  or  }
\log(m/n) > k + t\log t 
\label{lbndreq}
\end{equation}
 
 For $k = \log n$ and $m = n\cdot polylog(n) $, this holds for some $t \geq 
\frac{\alpha \log \log n}{\log\log\log n})$, i.e., for some constant 
$\alpha > 0$, i.e., 
the error cannot be
$o (\frac{\log\log n}{\log\log\log n})$. 
\end{proof}

\ignore{
\subsection{An alternate proof for the RT bound}
Our algorithm also provides an alternate proof for the 
$O(\frac{\log m}{\log\log m})$ error bound in constraints incurred in 
case of classical randomized rounding by Raghavan and 
Thompson \cite{RT:87}
We run our algorithm on the above problem(with scaling $S=1$).
As shown in the proof of Lemma \ref{absorp_rate},the expected number of 
unfixed variables in $i^{th}$ constraint $=E(u_i^p) (=\|V_i^p\|_2 ^2) \leq \frac{n}{2^p}$ (when $B = S = 1$).The above holds with high probability when $p \leq \log(n)-\log(\log(m))(=p*)$.\\
To simplify our analysis we define variables indexed by the number of phases after p* phases are over.
Let q indicate the number of phases after $p*$ phases are over.
Let $w_i^q=u_i^{q+\log (n)-\log (\log (m))}$ (that is the number of unfixed variables in $i^{th}$ constraint indexed by number of phases after p*)and let $\overset{\sim}{T_q}=T_{q+\log (n)-\log (\log (m))}$ (that is the number of iterations as a function of number of phases after p*).\\
Now $E(w_i^q)=\frac{n}{2^{\log (n)-\log (\log (m))+q}}=\frac{\log m}{2^q}$.
To get bound with sufficiently high probability,bound we need to run our algorithm for a sufficiently large number($>>\overset{\sim}{T_q}$) of steps so that,number of unabsorbed variables is bounded by $\frac{\log (m)}{2^q}$ with high probability.

We claim the following:-
\begin{lemma}
If r=$2^q \cdot 2^{2^q +1}$, then after $\overset{\sim}{T_r}=\frac{2^{r+\log(n)-\log(\log(m))} \cdot 2^{r+\log(n)-\log(\log(m))}}{n^2 \gamma^2}$ steps,the number of unfixed variables in $i^{th}$ equation is bounded by $\frac{\log(m)}{2^q}$ with probability $\geq 1-\frac{1}{m}$.
\end{lemma}
\begin{proof}
By definition $w_i^r$ is the number of unfixed variables in the $i^{th} constraint$ after $\overset{\sim}{T_r}$ iterations.\\
We run it for $\overset{\sim}{T_r}=\frac{2^{\log (n)-\log (\log (m))+r} 
2^{\log (n)-\log (\log (m))+r}}{n^2 \gamma^2}$ steps, where r is some function of q, so that $E(w_i^r)=\frac{\log (m)}{2^r}$\\
Thus,$Pr(w_i^r> \frac{\log (m)}{2^q})$\\
$=Pr(w_i^r>E(w_i^r)(1+\delta ))$, where $(1+\delta)\frac{\log (m)}{2^r}=\frac{
\log (m)}{2^q}$ \\
$\leq (\frac{e^\delta}{(1+\delta)^{1+\delta}})^{\frac{\log (m)}{2^r}}$\\
$=e^{\delta \frac{\log (m)}{2^r}-(1+\delta )\ln(1+\delta )\frac{\log (m)}{2^r}}$,  where $(1+\delta )\frac{\log (m)}{2^r}=\frac{\log (m)}{2^q}$\\
$=e^{\frac{\log (m)}{2^q}-\frac{\log (m)}{2^r}-\frac{\ln(1+\delta)\log (m)}{2^q}}  \text{  since } (1+\delta )\frac{\log (m)}{2^r}=\frac{\log (m)}{2^q} \text{   and  hence  }  \delta \frac{\log (m)}{2^r}=\frac{\log (m)}{2^q}-\frac{\log 
(m)}{2^r}$\\
$\leq e^{\frac{\log (m)}{2^q}(1-ln(1+\delta))}$\\
To bound the probability inverse polynomial in m,
we need to choose$\frac{\log (m)}{2^q}(ln(1+\delta)-1) \geq O(\log (m))$\\
or $\ln(1+\delta) \geq 2^q+1$\\
or $1+\delta  \geq 2^{2^q+1}$\\
or $ 2^r \geq 2^q 2^{2^q+1}$\\
Thus it suffices to run for $\overset{~}{T_r}=\frac{2^{r+\log (n)-\log 
(log (m))}.2^{r+\log (n)-\log (\log (m))}}{n^2 \gamma^2}$ iterations so as to have $w_i^r$ close to O($\frac{\log (m)}{2^q}$), with high probability, where r=$2^{2^q}.2^q$\\
\end{proof}

To prove a bound on the error, 
we proceed similar to the analysis in Section \ref{}. 
Consider a fixed constraint $C_j:V_j.\bar{x} \leq 1$.
\begin{lemma}For all constraints the error in $q^{th}$ phase after p* is 
$\overset{\sim}{\delta_q} \leq \delta_{q+\log(n)-\log(\log(m))}=2^{\frac{q}{2}} \cdot 2^{2^q+1}$
\begin{align*}
\text{Consider } &Pr(\vert < \gamma (\sum_{i=\overset{\sim}{T_{r-1}}+1}^{\overset{\sim}{T_{r}}} U_i),V_j>\vert \geq \beta_q)\\ 
&=Pr(\vert < \sum_{i=\overset{\sim}{T_{r-1}}+1}^{\overset{\sim}{T_{r}}} U_i, \frac{V_j}{\| V_j \|}> \vert \geq \frac{\beta_q}{\gamma \| V_j \|}).
\end{align*}
\end{lemma}
\begin{proof}
Now $<U_i,\frac{V_j}{\vert \vert V_j \vert \vert}> \sim \mathcal{N}(0, \sigma^2)$ where $\sigma^2 \leq 1$.\\
Thus the above is bounded by $e^{-\frac{\beta_q^2}{\gamma^2 \vert \vert V_j \vert \vert ^2 2 (\overset{\sim}{T_r}-\overset{\sim}{T_{r-1}})}}$
Thus for the $q+\log (n)-\log(\log (m))$ phase,we will need to choose $\beta_q$ satisfying:-
$\frac{\beta_q^2}{\gamma^2 \vert \vert V_j \vert \vert ^2 2 (\overset{\sim}{T_j}-\overset{\sim}{T_{j-1}})} \geq \log (m)$\\
i.e. if $\beta_q \geq \gamma \vert \vert V_j \vert \vert \sqrt{2.\log (m).(\overset{\sim}{T_j}-\overset{\sim}{T_{j-1}})}$\\
Since $\vert \vert V_j \vert \vert \leq \sqrt{\frac{\log n}{2^i}}$ and $\overset{\sim}{T_j}=\frac{2^{q+\log (n)-\log (\log (m))} 2^{2^q+1}.2^{q+\log (n)-
\log (\log (m)) 2^{2^q+1}}}{n^2 \gamma^2}$,\\ 
it suffices to choose \begin{align*}
 \beta_q &\geq 2^{q/2}.2^{2^q+1}
\end{align*}
For $q \leq \log (\log (\log (m)))-2$, the error is bounded by $2^{\frac{\log
(\log (\log (m)))-2}{2}}.2^{2^{\log (\log (\log (m)))-2}}=\sqrt{\log 
(m).\log (\log (m))}$\\
This leaves the number of unfixed random variables bounded by O($\frac{\log 
(m)}{2^q}=\frac{\log (m)}{\log (\log (m))}$).Increasing $q$ 
beyond this value will not give any advantage since the error will then exceed the error due obtained by setting variables to 1.
\end{proof}
}
\subsection{Proof of the  general approximation bound: Proof of Theorem \ref{mainthm3} }

If we are interested in maintaining feasibility of constraints, we can
ensure it by sacrificing the value of the objective function according to
some trade-offs. The same algorithm works, where we use the known method 
of damping the probabilities of rounding (\cite{RT:87,AS:95}).

Let $X_1 , X_2 \ldots X_u$ denote the unabsorbed random 
variables in a constraint
after the Brownian walk phase and let $ U = \sum_i X_i  = \frac{B}{\beta}$
for some $\beta \geq 1$. Let $\hat{X_i}$ be the $\{0, 1\}$ random variables
where $\Pr [ \hat{X_i} =1] = X_i$ and so $\E [ \sum_i  \hat{X_i} ] = 
\sum_i X_i = \frac{B}{\beta}$. To apply LLL, we need
$\Pr [ \sum_i \hat{X_i} > B ] < \frac{1}{d}$.

Since $U = \sum_i \hat{X_i}$, then from Chernoff bounds, we know that 
\begin{equation}
 \Pr [ U \geq (1+ \Delta)\E[ U ] \leq {\left[ \frac{e^\Delta }{ { (1+ 
\Delta )}^{ 1+ \Delta }} \right]}^{E[U]} 
\label{chernoff}
\end{equation}
Using $\E [ U ] = \frac{B}{\beta}$ and $\beta = 1+ \Delta $ we obtain
$\Pr [ U \geq B ] \leq {\left[ \frac{e^{\beta -1}}{ {(\beta )}^{\beta}}\right]
}^{B/\beta} \leq {\left( \frac{e}{\beta} \right)}^B \leq \frac{1}{\alpha d}$
for some scale parameter $\alpha \geq 1$.
It follows that 
\begin{equation}
 {( \beta /e )}^{B} > d\alpha 
\label{damp_req}
\end{equation}
Therefore $\beta = e {(d \alpha)}^{1/B}$.
Since our matrices are 0-1, the
minimum infeasibility happens at $B+1$, so we can substitute $B+1$ instead
of $B$ in the previous calculations to obtain $\beta = e {(d \alpha)}^{1/B+1}$.  
Using $S \geq 1$ as the initial scaling for Brownian walk,
we obtain the following result using Lemma \ref{err_bnd}

\begin{lemma}
For $B \geq 1$,
we can round the fractional solution to a feasible integral solution
with objective function value $\Omega ( OPT /S )$ for S=max(${(\frac{m}{OPT}})^\frac{1}{B-1},d^\frac{1}{B-1}$).
If $B \in \mathbb{Z}$,B can be replaced by B+1 in above exprssion and proof 
\label{colbndapprox}
\end{lemma}
\begin{proof}
If we start from $\overset{\sim}{x}=\frac{x'}{S}$, the solution produced by brownian walk satisfies the constraints with RHS  
$\frac{B}{\beta}$(as shown in Lemma $\ref{err_bnd}$).\\
Now we define events for LLL similar to the proof of Lemma \ref{bnd_col}.\\
$\forall 1 \leq i \leq m$ define $E_i\equiv (A_i^T.x >B)$.\\
and define $E_{m+1}\equiv c^T.x <(1-\epsilon).OPT$.\\
As shown above, $\forall 1\leq i \leq m$, $Pr(E_i)\leq (\frac{e}{\beta})^B$.\\
Also by chernoff's bound $Pr(E_{m+1})\leq e^{{-\epsilon^2}{OPT/2}}$.\\
We apply LLL with weights $y_i=\frac{1}{\alpha d}$, for $1\leq i \leq m$\\
and $y_{m+1}=\frac{1}{2}$.\\,
for some appropriately defined $\alpha \geq 1$.\\
Now we need $(\frac{e}{\beta})^B \leq \frac{1}{\alpha.d}(1-\frac{1}{\alpha d})^d.\frac{1}{2}$\\
and $e^{\frac{-\epsilon^2 OPT}{2 \beta}} \leq \frac{1}{2}.(1-\frac{1}{\alpha .d})^m \leq \frac{1}{2}.e^{-\frac{m}{\alpha.d}}$(From Equation (\ref{eqnobj}))\\
The above is satisfied for $\frac{OPT}{\beta}\geq \frac{m}{\alpha.d}$\\
Thus we need to choose $\frac{OPT}{(\alpha.d)^\frac{1}{B}}\geq \frac{m}{\alpha.d}$\\
OR $\beta=(\alpha.d)^{1/B} \geq (\frac{m}{OPT})^{\frac{1}{B-1}}$.\\
Also $\alpha \geq 1$.For $B \in \mathbb{Z}$, it would have sufficed to use $B'=B+1$ as opposed to choosing B which proves the stated lemma.
\end{proof}
To prove Theorem \ref{mainthm3}, we can observe that for $c_i \geq \frac{1}{p}$, we have $OPT \geq \frac{n}{kp}$.
For $\frac{mk}{n}=O(\log(m))$, we get the scaling as $max((\log(m))^\frac{2}{B},(p \log(m))^\frac{1}{B})$, which proves Theorem \ref{mainthm3}\\

\color{black}

\end{document}